\documentclass{article}
\usepackage{a4wide}
\usepackage{amsmath}
\usepackage{amssymb}
\usepackage{amsthm}

\DeclareMathOperator{\diag}{diag}
\newtheorem{thm}{Theorem}

\newtheorem{remark}{Remark}

\newtheorem{lemma}{Lemma}

\newcommand{\bbr}{{\mathbb R}}
\newcommand{\bbs}{{\mathbb S}}

\newcommand{\p}{\hat{p}}
\newcommand{\q}{\hat{q}}
\newcommand{\n}{\hat{n}}

\usepackage{authblk}

\begin{document}
\title{Late-time behaviour of Israel particles in a FLRW spacetime with $\Lambda>0$}

\author[1]{Ho Lee\footnote{holee@khu.ac.kr}}
\author[2]{Ernesto Nungesser\footnote{ernesto.nungesser@icmat.es}}

\affil[1]{Department of Mathematics and Research Institute for Basic Science, Kyung Hee University, Seoul, 130-701, Republic of Korea}
\affil[2]{Instituto de Ciencias Matem\'{a}ticas (CSIC-UAM-UC3M-UCM), 28049 Madrid, Spain}

\maketitle

\begin{abstract}
In this paper we study the relativistic Boltzmann equation in a spatially flat FLRW spacetime. We consider Israel particles, which are the relativistic counterpart of the Maxwellian particles, and obtain global-in-time existence and the asymptotic behaviour of solutions. The main argument of the paper is to use the energy method of Guo, and we observe that the method can be applied to study small solutions in a cosmological case. It is the first result of this type where a physically well-motivated scattering kernel is considered for the general relativistic Boltzmann equation.
\end{abstract}

\section{Introduction}\label{sec_introduction}
The Einstein-Boltzmann system is a system of equations which can describe the time evolution of a collection of particles, where the particles interact with each other via gravitation and collisions. In this paper the equations will be studied in the cosmological context, and a spatially flat Friedman-Lema\^{\i}tre-Robertson-Walker(FLRW) spacetime with a positive cosmological constant will be considered. The main interest of the present paper is to understand the late-time behaviour of matter distribution, rather than the evolution of spacetime itself, and we refer to \cite{CercignaniKremer,Choquet,Hans} for general relativity or relativistic kinetic theories.

There are not so many results concerning solutions to the Einstein-Boltzmann system. Local existence of solutions was shown in \cite{Bancel1,Bancel2}, but it was not until much later that global problems were studied by Noutchegueme and his collaborators. For instance, a global existence result for spatially flat FLRW spacetimes with the cosmological constant was obtained in \cite{NTRW}, where it was also shown that there cannot exist global solutions with a negative cosmological constant. This result was generalised to the Bianchi I LRS case in \cite{ND}. On the other hand, the asymptotic behaviour of solutions to the Boltzmann equation was recently obtained in \cite{Lee13} with FLRW symmetry. This result has been extended to the general Bianchi I case in \cite{LeeNungesser15,LeeNungesser16} for particles with and without mass.

The feature of the present paper is twofold. First, we do not impose any artificial restrictions on the scattering kernel. In the results mentioned above certain artificial restrictions were imposed for technical reasons, which are not physically motivated, for instance boundedness or restrictions on the angular variable. In this paper we obtain global existence of solutions and the asymptotic behaviour in the case of a positive cosmological constant. In particular we focus now on Israel particles \cite{Israel63} which are the relativistic counterpart of the Maxwellian particles. Second, we use a new representation of post-collision momenta to obtain classical solutions. For the relativistic Boltzmann equation it was observed in \cite{GuoStrain12} that two different representations for post-collision momenta must be considered to obtain classical solutions, and this was applied to the Bianchi I case in \cite{LeeNungesser15}. In this paper we derive a new representation and obtain classical solutions with this single representation. More precisely, we will observe that partial derivatives of post-collision momenta in the new representation do not have singularities, but are uniformly bounded in $p$.

\section{FLRW spacetimes}\label{sec_flrw}
The FLRW spacetimes are the homogeneous and isotropic cosmological models, and we will assume that the metric is given by
\[
^{(4)}g=-dt^2+g,\quad g=R^2\eta,\quad \eta=\diag(1,1,1),
\]
where $R=R(t)$ is the scale factor. Note that later $\eta_{\alpha\beta}$ will be used to denote the Minkowski metric $\eta_{\alpha\beta}=\diag(-1,1,1,1)$. Denoting by a dot the derivative with respect to time, the governing equations in a FLRW spacetime with the cosmological constant $\Lambda$ are given as follows:
\begin{align}\label{Friedmann}
\frac{{\dot{R}}^{2}}{R^{2}}&=\frac{8\pi\rho+\Lambda}{3},\\
\frac{3\ddot{R}}{R}&=-4\pi(\rho+3P)+\Lambda \label{f2},
\end{align}
where $\rho$ and $P$ are the energy density and the pressure, respectively. In order to show the existence of solutions of the Boltzmann equation we will need the following lemma.

\begin{lemma}\label{lem_flrw}
Consider an initially expanding FLRW spacetime with a positive cosmological constant $\Lambda$, which satisfies the weak and the dominant energy conditions. Then, the scale factor grows exponentially with time and $\Lambda$.
\end{lemma}
\begin{proof}
We multiply \eqref{Friedmann} by $R^2$ and derive with respect to time to obtain
\begin{align}
2\dot{R}\ddot{R}=\frac{8\pi\dot{\rho}}{3}R^2+\frac{8\pi\rho+\Lambda}{3}2\dot{R}R.
\end{align}
If we substitute $\ddot{R}$ with \eqref{f2}, then we obtain
\begin{align}
\frac23 \dot{R}(-4\pi(\rho+3P)+\Lambda)R=\frac{8\pi\dot{\rho}}{3}R^2+\frac{8\pi\rho+\Lambda}{3}2\dot{R}R,
\end{align}
which can be simplified to
\begin{align}\label{en}
\dot{\rho}=-\frac{3\dot{R}}{R}(\rho+P).
\end{align}
For any matter which satisfies the weak energy condition we have $\rho\geq 0$. In the presence of a positive cosmological constant this implies that $\dot{R}^2>0$. We are interested in spacetimes that are initially expanding, i.e., we will assume
\begin{align}
\dot{R}(0)>0, \quad R(0)=1,
\end{align}
where the second equality can always be obtained by rescaling if necessary. Due to continuity this implies that our universe expands forever, and we ignore from now on the negative root of \eqref{Friedmann} which corresponds to initially contracting spacetimes. For matter which satisfies the dominant energy condition we have also that the pressure is bounded by the energy density, which means that the energy density is non-increasing due to \eqref{en}.

In fact from what has been said we have the following bounds
\begin{align}
C_1=\frac{\Lambda}{3} \leq \frac{{\dot{R}}^{2}}{R^{2}} \leq \frac{8\pi\rho(0)+\Lambda}{3}=C_2,
\end{align}
which implies that
\begin{align}
e^{\sqrt{C_1} t}  \leq R(t) \leq e^{\sqrt{C_2} t},
\end{align}
where the equalities hold in case of vacuum, which are the well-known de Sitter spacetimes. We see thus that for any matter which satisfies the usual energy conditions the scale factor grows exponentially.
\end{proof}

In this paper we will assume that an initially expanding FLRW spacetime with $\Lambda>0$ is given and will study the late-time behaviour of the Boltzmann equation. One may consider a coupled system such as the Einstein-Boltzmann system. In this case, the equations \eqref{Friedmann}--\eqref{f2} are coupled to the Boltzmann equation \eqref{Boltzmann} through the energy-momentum tensor, and this will be discussed in Section \ref{sec_summary}.

\subsection{Notations}
We use the following notation for partial derivatives: for a multi-index $I=(i,j,k)$ with $i$, $j$, and $k$ non-negative integers, $\partial_I =(\partial^1)^i(\partial^2)^j(\partial^3)^k$, where $\partial^a=\partial /\partial p_a$ is the partial derivative with respect to $p_a$ for each $a=1,2,3$. Note that $\partial^a=g^{a\beta}\partial_\beta =R^{-2}\partial_a$ in the FLRW case, where $\partial_a=\partial /\partial p^a$. Greek letters run from $0$ to $3$, while Latin letters from $1$ to $3$, and as usual we assume the Einstein summation convention. Momentum variables without indices will denote three dimensional vectors.

The orthonormal frame approach will be frequently used to simplify calculations. For instance, we have the relation $\p^a=p_a/R$, where the hat denotes that the momentum is written in an orthonormal frame. In an orthonormal frame the Minkowski metric will apply, i.e., $\p_a=\eta_{ab}\p^b$, and the quantities like $\q$ and $\n$ will be understood in a similar way. For partial derivatives in an orthonormal frame, we use hat as $\hat{\partial}_a=\partial /\partial \hat{p}^a$ and have $\partial^a=R^{-1}\hat{\partial}_a$ for each $a=1,2,3$. For a multi-index $I=(i,j,k)$, the operator $\hat{\partial}_I$ will be understood as $\hat{\partial}_I=(\hat{\partial}_1)^i(\hat{\partial}_2)^j(\hat{\partial}_3)^k$. For simplicity, we use the following notations:
\[
|p_*|^2=\sum_{i=1}^3(p_i)^2,\quad |p|^2=\sum_{i=1}^3(p^i)^2,\quad|\hat{p}|^2=\sum_{i=1}^3(\p^i)^2.
\]
Note that $|\p|^2=\eta_{ab}\p^a\p^b$, but $|p_*|^2\neq g^{ab}p_ap_b$ and $|p|^2\neq g_{ab}p^ap^b$. With these notations we define the weight function:
\[
\langle p_*\rangle=\sqrt{1+|p_*|^2},
\]
and note that $\langle p_*\rangle\neq p^0=\sqrt{1+R^{-2}|p_*|^2}$.

\section{The Boltzmann equation in a spatially flat FLRW spacetime}\label{sec_boltzmann}
In a spatially flat FLRW spacetime the Boltzmann equation is written as
\[
\partial_tf=Q(f,f),
\]
where the distribution function $f=f(t,p_*)$ is understood as a function of time $t$ and covariant variables $p_*=(p_1,p_2,p_3)$. To determine the collision operator $Q$, let us consider an orthonormal frame $\{e_\mu\}$, and obtain $Q$ with the representation of \eqref{Q_ortho} and the post-collision momentum \eqref{p'_ortho}. To revert to a coordinate frame, we write $e_\mu=e^\alpha_\mu\partial/\partial x^\alpha$ and $p^\alpha=\p^\mu e^\alpha_\mu$. Since $e^\alpha_\mu e^\beta_\nu g_{\alpha\beta}=\eta_{\mu\nu}$, we have $\p^\mu\eta_{\mu\nu}=\p^\mu e^\alpha_\mu e^\beta_\nu g_{\alpha\beta}=p^\alpha e^\beta_\nu g_{\alpha\beta}=p_\alpha e^\alpha_\nu$ and $\p^\mu=p_\alpha e^\alpha_\nu\eta^{\mu\nu}$. Then, the quantities in \eqref{p'_ortho} can be written in terms of covariant variables. For instance, $\n\cdot\omega=\eta_{ab}\n^a\omega^b=\eta_{ab}n_\alpha e^\alpha_\nu\eta^{a\nu}\omega^b=n_\alpha e^\alpha_b\omega^b$.
In a spatially flat FLRW case, we can introduce an explicit form of an orthonormal frame as follows: $e^0_0=1$, $e^a_a=R^{-1}$, and $e^\alpha_\mu=0$ otherwise. Then, we have $\n\cdot\omega=n_\alpha e^\alpha_b\omega^b=n_a\omega^a/R$ and observe that $\p^0=-p_0=p^0$ and $\p^a=p_a/R$ for each $a=1,2,3$. The Boltzmann equation is now written as
\begin{align}
\partial_tf=R^{-3}\int_{\bbr^3}\int_{\bbs^2} v_M\sigma(h,\theta)\Big(f(p_*')f(q_*')-f(p_*)f(q_*)\Big)d\omega dq_*,\label{Boltzmann}
\end{align}
where $v_M$ is the M{\o}ller velocity defined by
\[
v_M=\frac{h\sqrt{s}}{4p^0q^0},\quad h=\sqrt{(p_\alpha-q_\alpha)(p^\alpha-q^\alpha)},\quad s=-(p_\alpha+q_\alpha)(p^\alpha+q^\alpha),
\]
where $h$ and $s$ are called the relative momentum and the total energy, respectively. The post-collision momentum $p_\alpha'$ is now given by
\begin{align}
\left(
\begin{array}{c}
p'^0\\
p'_k
\end{array}
\right)=
\left(
\begin{array}{c}
\displaystyle
p^0+2\bigg(-q^0\frac{n_a\omega^a}{\sqrt{s}}+q_a\omega^a+\frac{n_a\omega^an_b q^b}{\sqrt{s}(n^0+\sqrt{s})}\bigg)\frac{n_c\omega^c}{R^2\sqrt{s}}\\
\displaystyle
p_k+2\bigg(-q^0\frac{n_a\omega^a}{\sqrt{s}}+q_a\omega^a+\frac{n_a\omega^an_bq^b}{\sqrt{s}(n^0+\sqrt{s})}\bigg)
\bigg(\omega^c\eta_{ck}+\frac{n_c\omega^cn_k}{R^2\sqrt{s}(n^0+\sqrt{s})}\bigg)
\end{array}
\right),\label{p'}
\end{align}
and $q_\alpha'=p_\alpha+q_\alpha-p_\alpha'$, where $n^\alpha$ denotes $p^\alpha+q^\alpha$ for simplicity, and $\omega=(\omega^1,\omega^2,\omega^3)\in\bbs^2$ serves as a parameter. For more details we refer to the appendix.

\subsection{Israel particles}
In this paper we are interested in the Israel particles \cite{Israel63}, i.e., the scattering kernel is given by
\[
\sigma(h,\theta)=\frac{4}{h(4+h^2)}\sigma_0(\theta),\label{scat}
\]
where $\sigma_0$ is an arbitrary function of the scattering angle $\theta$ (see (5.101) of \cite{CercignaniKremer} for more details). This corresponds to the Maxwellian particles in the non-relativistic limit, and for simplicity we assume that
\[
\sigma_0(\theta)\equiv 1.
\]
Hence, the scattering kernel of our interest is written as $\sigma(h,\theta)=4(hs)^{-1}$.

\section{Preliminaries}\label{sec_preliminaries}
In this part we collect basic estimates.

\begin{lemma}\label{lem_basic}
The following estimates hold:
\begin{align*}
s=4+h^2,\quad \frac{|\p-\q|}{\sqrt{p^0q^0}}\leq h\leq |\p-\q|,\quad
s\leq 4p^0q^0,\quad
|\p|=R^{-1}|p_*|\leq p^0.
\end{align*}
\end{lemma}
\begin{proof}
These estimates follow easily from the definitions and assumptions. We refer to \cite{GuoStrain12,LeeNungesser15} for the proofs.
\end{proof}

\begin{lemma}\label{lem_R}
For any integer $m$, we have
\[
\int_{\bbr^3}(p^0)^me^{-p^0}dp_*\leq CR^3,
\]
where the constant $C$ does not depend on $t$.
\end{lemma}
\begin{proof}
This is a simple calculation given by
\begin{align*}
&\int_{\bbr^3}(p^0)^me^{-p^0}dp_*=\int_{\bbr^3} (1+R^{-2}|p_*|^2)^{\frac{m}{2}}e^{-\sqrt{1+R^{-2}|p_*|^2}}dp_*\\
&=R^3\int_{\bbr^3}(1+|z|^2)^{\frac{m}{2}}e^{-\sqrt{1+|z|^2}}dz\leq CR^3,
\end{align*}
and this completes the proof.
\end{proof}

\begin{lemma}\label{lem_pp'q'}
Suppose that $p_*'$ and $q_*'$ are represented by \eqref{p'} for given $p_*$ and $q_*$. Then, we have
\[
\langle p_*\rangle\leq C\langle p_*'\rangle\langle q_*'\rangle,
\]
where the constant $C$ does not depend on $t$.
\end{lemma}
\begin{proof}
In an orthonormal frame we have
\[
\frac{1+|p_*|^2}{(1+|p_*'|^2)(1+|q_*'|^2)}=\frac{1+R^2|\p|^2}{(1+R^2|\p'|^2)(1+R^2|\q'|^2)}\leq \frac{1+R^2|\p|^2}{1+R^2(|\p'|^2+|\q'|^2)},
\]
and will show that this quantity is bounded uniformly in time. In the first case, ${\sf (i)}$ $|\p|\geq 4$, we have Lemma 2.2 of \cite{GlasseyStrauss93}, which shows that $\max\{|\p'|,|\q'|\}\geq |\p|/4$. Hence, we have
\[
\frac{1+R^2|\p|^2}{1+R^2(|\p'|^2+|\q'|^2)}\leq \frac{1+R^2|\p|^2}{1+R^2|\p|^2/16}\leq 16.
\]
In the second case, ${\sf (ii)}$ $|\p|\leq 4$ and $|\q|\geq 4$, we use the energy conservation:
\[
\sqrt{1+|\p'|^2}+\sqrt{1+|\q'|^2}=\sqrt{1+|\p|^2}+\sqrt{1+|\q|^2}.
\]
This shows that $4\leq |\q|\leq\sqrt{1+|\q|^2}\leq \sqrt{1+|\p'|^2}+\sqrt{1+|\q'|^2}$, which implies
\[
\max\{\sqrt{1+|\p'|^2},\sqrt{1+|\q'|^2}\}\geq 2,
\]
and equivalently we have $\max\{|\p'|^2,|\q'|^2\}\geq 3$. Hence, we have
\[
\frac{1+R^2|\p|^2}{1+R^2(|\p'|^2+|\q'|^2)}\leq \frac{1+16R^2}{1+3R^2}\leq \frac{16}{3}.
\]
In the last case, ${\sf (iii)}$ $|\p|\leq 4$ and $|\q|\leq 4$, we use the representation \eqref{p'_alpha} such that
\[
\p'=\p-((\p_\alpha-\q_\alpha)\Omega^\alpha)\Omega,\quad\q'=\q+((\p_\alpha-\q_\alpha)\Omega^\alpha)\Omega.
\]
For simplicity let us write $U^\alpha=\p^\alpha-\q^\alpha$. Then, we have
\begin{align*}
|\p'|^2+|\q'|^2&=\bigg|\frac{\p+\q}{2}+\bigg(\frac{\p-\q}{2}-(U_\alpha\Omega^\alpha)\Omega\bigg)\bigg|^2 +\bigg|\frac{\p+\q}{2}-\bigg(\frac{\p-\q}{2}-(U_\alpha\Omega^\alpha)\Omega\bigg)\bigg|^2\\
&=\frac{|\p+\q|^2}{2}+\frac{|U-2(U_\alpha\Omega^\alpha)\Omega|^2}{2}.
\end{align*}
Since $\Omega_\alpha \Omega^\alpha=1$ and $U_\alpha U^\alpha=h^2$, we have
\begin{align*}
h^2&=U_\alpha U^\alpha=(U_\alpha-2(U_\beta\Omega^\beta)\Omega_\alpha)(U^\alpha-2(U_\beta\Omega^\beta)\Omega^\alpha)\\
&=-(U^0-2(U_\beta\Omega^\beta)\Omega^0)^2+|U-2(U_\beta\Omega^\beta)\Omega|^2\leq |U-2(U_\beta\Omega^\beta)\Omega|^2.
\end{align*}
Hence, we obtain
\[
|\p'|^2+|\q'|^2\geq \frac{|\p+\q|^2}{2}+\frac{h^2}{2}.
\]
We now apply Lemma \ref{lem_basic} to estimate the right side as follows:
\[
\frac{|\p+\q|^2}{2}+\frac{h^2}{2}\geq \frac{|\p+\q|^2}{2}+\frac{|\p-\q|^2}{2p^0q^0}\geq\frac{|\p+\q|^2+|\p-\q|^2}{2p^0q^0}\geq \frac{|\p|^2+|\q|^2}{17},
\]
and this shows that $|\p'|^2+|\q'|^2\geq (|\p|^2+|\q|^2)/17$. Hence, we have
\[
\frac{1+R^2|\p|^2}{1+R^2(|\p'|^2+|\q'|^2)}\leq \frac{1+R^2|\p|^2}{1+R^2|\p|^2/17}\leq 17,
\]
and this completes the proof of the lemma.
\end{proof}

\begin{lemma}\label{lem_partial_1/p0}
For a multi-index $I$, there exists a polynomial ${\mathcal P}$ such that
\[
\partial_I \bigg[\frac{1}{p^0}\bigg]=\frac{1}{R^{|I|}(p^0)^{|I|+1}}{\mathcal P}\bigg(\frac{p_*}{Rp^0}\bigg),
\]
where $p_*=(p_1,p_2,p_3)$.
\end{lemma}
\begin{proof}
We prove this lemma in an orthonormal frame, and then use the relation $\partial^a=R^{-1}\hat{\partial}_a$ to conclude the lemma. We claim that for a multi-index $I$, there exists a polynomial ${\mathcal P}$ such that
\begin{align}
\hat{\partial}_I \bigg[\frac{1}{p^0}\bigg]=\frac{1}{(p^0)^{|I|+1}}{\mathcal P}\bigg(\frac{\p}{p^0}\bigg).\label{partial_1/p0}
\end{align}
The proof is given by an induction. Note that \eqref{partial_1/p0} holds trivially for $|I|=0$, and then suppose that \eqref{partial_1/p0} holds for some $I$ such that $|I|=n\geq 0$. We now prove \eqref{partial_1/p0} for $I'$ such that $|I'|=n+1$. We first notice that
\begin{align}
\hat{\partial}_ap^0=\frac{\hat{p}_a}{p^0},\quad \hat{\partial}_b\bigg[\frac{\hat{p}_a}{p^0}\bigg]=\frac{1}{p^0}\bigg(\eta_{ab}-\frac{\p_a}{p^0}\frac{\p_b}{p^0}\bigg)=\frac{1}{p^0} P_1\bigg(\frac{\p}{p^0}\bigg), \label{partial_p0}
\end{align}
for some polynomial $P_1$. This shows that for any polynomial $P$,
\begin{align}
\hat{\partial}_a\bigg[P\bigg(\frac{\p}{p^0}\bigg)\bigg]=\frac{1}{p^0}P_2\bigg(\frac{\p}{p^0}\bigg),\label{partial_P}
\end{align}
for some polynomial $P_2$. Since $\hat{\partial}_{I'}=\hat{\partial}_a\hat{\partial}_{I}$ for some $a$ and $|I|=n$, we have
\begin{align*}
\hat{\partial}_{I'}\bigg[\frac{1}{p^0}\bigg]&=\frac{\partial}{\partial \p^a}\bigg[\frac{1}{(p^0)^{n+1}}P\bigg(\frac{\p}{p^0}\bigg)\bigg]\\
&=\frac{-(1+n)}{(p^0)^{n+2}}\frac{\p_a}{p^0}P\bigg(\frac{\p}{p^0}\bigg)+\frac{1}{(p^0)^{n+1}}\frac{\partial}{\partial \p^a}\bigg[P\bigg(\frac{\p}{p^0}\bigg)\bigg]\allowdisplaybreaks\\
&=\frac{1}{(p^0)^{n+2}}P_3\bigg(\frac{\p}{p^0}\bigg)=\frac{1}{(p^0)^{|I'|+1}}P_3\bigg(\frac{\p}{p^0}\bigg),
\end{align*}
where $P_3$ denotes another polynomial. This proves \eqref{partial_1/p0}, and the relations $\partial^a=R^{-1}\hat{\partial}_a$ and $p_*=R\p$ derive the desired result.
\end{proof}

\begin{lemma}\label{lem_partial_1/s}
For a multi-index $I$, there exist polynomials ${\mathcal P}_i$ such that
\[
\partial_I\bigg[\frac{1}{\sqrt{s}}\bigg]=\frac{1}{R^{|I|}\sqrt{s}}\sum_{i=0}^{|I|}\bigg(\frac{q^0}{s}\bigg)^i \bigg(\frac{1}{p^0}\bigg)^{|I|-i}{\mathcal P}_i\bigg(\frac{p_*}{Rp^0},\frac{q_*}{Rq^0}\bigg),
\]
where $p_*=(p_1,p_2,p_3)$.
\end{lemma}
\begin{proof}
We first show that in an orthonormal frame there exist some polynomials ${\mathcal P}_i$ such that
\begin{align}
\hat{\partial}_I\bigg[\frac{1}{\sqrt{s}}\bigg]=\frac{1}{\sqrt{s}}\sum_{i=0}^{|I|}\bigg(\frac{q^0}{s}\bigg)^i \bigg(\frac{1}{p^0}\bigg)^{|I|-i}{\mathcal P}_i\bigg(\frac{\p}{p^0},\frac{\q}{q^0}\bigg),\label{partial_1/s}
\end{align}
which holds trivially for $|I|=0$. Suppose that it holds for some $I$ such that $|I|=n\geq 0$. By a direct calculation, with $s=2+2p^0q^0-2\eta_{ab}\p^a\q^b$, we have
\begin{align}
\hat{\partial}_as=2(\hat{\partial}_ap^0)q^0-2\q_a=q^0P_1\bigg(\frac{\p}{p^0},\frac{\q}{q^0}\bigg),\label{partial_s}
\end{align}
where we used \eqref{partial_p0} and $P_1$ denotes a polynomial. Then, we obtain
\begin{align*}
&\hat{\partial}_a\bigg[\frac{1}{\sqrt{s}}\bigg]=-\frac{\hat{\partial}_as}{2s^{3/2}}=\frac{1}{\sqrt{s}}\bigg(\frac{q^0}{s}\bigg)P_2, \quad \hat{\partial}_a\bigg[\frac{q^0}{s}\bigg]=-\frac{q^0\hat{\partial}_as}{s^2}=\bigg(\frac{q^0}{s}\bigg)^2P_3,
\end{align*}
where $P_2$ and $P_3$ are some polynomials having $\p/p^0$ and $\q/q^0$ as variables. These calculations, together with the previous lemma and the estimate \eqref{partial_P}, show that \eqref{partial_1/s} holds for $I'$ satisfying $|I'|=n+1$. Applying to \eqref{partial_1/s} the relations $\partial^a=R^{-1}\hat{\partial}_a$, $p_*=R\p$, and $q_*=R\q$, we obtain the desired result.
\end{proof}

\begin{lemma}\label{lem_partial_1/n0+s}
For a multi-index $I\neq 0$, there exist polynomials ${\mathcal P}_i$ such that
\[
\partial_I\bigg[\frac{1}{n^0+\sqrt{s}}\bigg]=\frac{1}{R^{|I|}}\sum_{i=1}^{|I|}\frac{(q^0)^{|I|}}{(n^0+\sqrt{s})^{i+1}} {\mathcal P}_i\bigg(\frac{p_*}{Rp^0},\frac{q_*}{Rq^0}, \frac{1}{p^0}, \frac{1}{q^0},\frac{1}{\sqrt{s}}\bigg),
\]
where $p_*=(p_1,p_2,p_3)$.
\end{lemma}
\begin{proof}
In this lemma we consider a multi-index $I$ such that $|I|\geq 1$. As in the previous lemmas, we prove the lemma in an orthonormal frame:
\begin{align}
\hat{\partial}_I\bigg[\frac{1}{n^0+\sqrt{s}}\bigg]=\sum_{i=1}^{|I|}\frac{(q^0)^{|I|}}{(n^0+\sqrt{s})^{i+1}} {\mathcal P}_i\bigg(\frac{\p}{p^0},\frac{\q}{q^0}, \frac{1}{p^0}, \frac{1}{q^0},\frac{1}{\sqrt{s}}\bigg).\label{partial_1/n0+s}
\end{align}
Below, $P_i$ will denote some polynomials of $\p/p^0$, $\q/q^0$, $1/p^0$, $1/q^0$, and $1/\sqrt{s}$. To get an induction we first consider $\hat{\partial}_I=\hat{\partial}_a$ to obtain
\[
\hat{\partial}_a\bigg[\frac{1}{n^0+\sqrt{s}}\bigg]=\frac{-(\hat{\partial}_ap^0+\hat{\partial}_a \sqrt{s})}{(n^0+\sqrt{s})^2}=\frac{P_1+q^0P_2}{(n^0+\sqrt{s})^2}=\frac{q^0}{(n^0+\sqrt{s})^2}P_3,
\]
where we used \eqref{partial_p0} and \eqref{partial_1/s}, and observe that \eqref{partial_1/n0+s} holds for $|I|=1$. The representations \eqref{partial_1/p0}, \eqref{partial_p0}, and \eqref{partial_1/s} show that for any polynomial $P$,
\[
\hat{\partial}_b\bigg[P\bigg(\frac{\p}{p^0},\frac{\q}{q^0}, \frac{1}{p^0}, \frac{1}{q^0},\frac{1}{\sqrt{s}}\bigg)\bigg] = q^0 P_4\bigg(\frac{\p}{p^0},\frac{\q}{q^0}, \frac{1}{p^0}, \frac{1}{q^0},\frac{1}{\sqrt{s}}\bigg),
\]
for some polynomial $P_4$. It is now easy to see that \eqref{partial_1/n0+s} holds for any $|I|=n$. The lemma is now obtained by applying the relations $\partial^a=R^{-1}\hat{\partial}_a$, $p_*=R\p$, and $q_*=R\q$.
\end{proof}

\begin{lemma}\label{lem_p'}
Consider post-collision momenta $p_*'$ and $q_*'$ in \eqref{p'} for given $p_*$ and $q_*$. Then, high order derivatives of them are estimated as follows: for a multi-index $I\neq 0$,
\[
|\partial_Ip_*'|+|\partial_Iq_*'|\leq CR^{1-|I|}(q^0)^{|I|+4},
\]
where the constant $C$ does not depend on $p_*$.
\end{lemma}
\begin{proof}
By the relations $\partial_I=R^{-|I|}\hat{\partial}_I$ and $p'_*=R\p'$, we have $\partial_Ip'_*=R^{1-|I|}\hat{\partial}_I\p'$. Hence, we will obtain an estimate for $\hat{\partial}_I\p'$, and will conclude the lemma by applying these relations. The representation \eqref{p'} can be written in an orthonormal frame as in \eqref{p'_alpha} with $\Omega^\alpha$:
\[
\p'^\alpha=\p^\alpha+2(\q_\beta\Omega^\beta)\Omega^\alpha,\quad \Omega^\alpha=\bigg(\frac{(\hat{n}\cdot\omega)}{\sqrt{s}},\omega+\frac{(\n\cdot\omega)\n}{\sqrt{s}(n^0+\sqrt{s})}\bigg),
\]
which shows that it is enough to consider the derivatives of $(\n\cdot\omega)/\sqrt{s}$ and $\n/(n^0+\sqrt{s})$. Below, the following estimate of $1/\sqrt{s}$ will be frequently used: if $|\p|\geq\max\{2|\q|,1\}$, then
\begin{align}\label{est_1/s}
\frac{1}{\sqrt{s}}\leq\frac{1}{h}\leq\frac{\sqrt{p^0q^0}}{|\p-\q|}\leq \frac{2\sqrt{p^0q^0}}{|\p|}\leq 4\sqrt{\frac{q^0}{p^0}},
\end{align}
and if $|\p|\leq\max\{2|\q|,1\}$, then $1/\sqrt{s}\leq C$ by Lemma \ref{lem_basic}.

We first estimate the quantity $\hat{\partial}_I((\n\cdot\omega)/\sqrt{s})$. For $|I|=0$, we have
\begin{align}\label{est_Omega1}
\bigg|\hat{\partial}_I\bigg[\frac{(\hat{n}\cdot\omega)}{\sqrt{s}}\bigg]\bigg|\leq \frac{|\n|}{\sqrt{s}}\leq
\left\{
\begin{aligned}
Cq^0,\quad\mbox{if}\quad |\p|\leq\max\{2|\q|,1\},\\
C(p^0)^{\frac12}(q^0)^{\frac12},\quad\mbox{if}\quad |\p|\geq\max\{2|\q|,1\},
\end{aligned}
\right.
\end{align}
where we used $|\n|\leq |\p|+|\q|\leq Cq^0$ in the first estimate, and $|\n|\leq |\p|+|\q|\leq Cp^0$ in the second estimate with \eqref{est_1/s}. High order derivatives of $(\n\cdot\omega)/\sqrt{s}$ are written as follows:
\begin{align}\label{est_Omega1-1}
\hat{\partial}_I\bigg[\frac{(\n\cdot\omega)}{\sqrt{s}}\bigg]=\sum \hat{\partial}_{I_1}\Big[\n\cdot\omega\Big]\hat{\partial}_{I_2}\bigg[\frac{1}{\sqrt{s}}\bigg],
\end{align}
which is a finite sum for some multi-indices $I_1$ and $I_2$. Note that the quantity $\hat{\partial}_{I_1}(\n\cdot\omega)$ is bounded by a constant $C$ or by $|\n|$ for $|I_1|\geq 1$ or $|I_1|=0$, respectively, and that the representation \eqref{partial_1/s} shows that
\begin{align}\label{est_Omega1-2}
\bigg|\hat{\partial}_{I_2}\bigg[\frac{1}{\sqrt{s}}\bigg]\bigg|\leq C\frac{(q^0)^{|I_2|}}{\sqrt{s}}\sum_{i=0}^{|I_2|}\bigg(\frac{1}{s}\bigg)^i\bigg(\frac{1}{p^0}\bigg)^{|I_2|-i} \leq C\frac{(q^0)^{|I_2|}}{\sqrt{s}},
\end{align}
for any $|I_2|\geq 0$. If $|I_2|\geq 1$, then the above estimate is improved by
\begin{align}\label{est_Omega1-3}
\bigg|\hat{\partial}_{I_2}\bigg[\frac{1}{\sqrt{s}}\bigg]\bigg|\leq C\frac{(q^0)^{|I_2|}}{\sqrt{s}}\sum_{i=0}^{|I_2|}\bigg(\frac{1}{s}\bigg)^i\bigg(\frac{1}{p^0}\bigg)^{|I_2|-i} \leq C\frac{(q^0)^{|I_2|}}{\sqrt{s}}\bigg(\frac{1}{s}+\frac{1}{p^0}\bigg).
\end{align}
We now consider \eqref{est_Omega1-1} for $|I|\geq 1$. If $|I_1|\geq 1$, then the quantity $\hat{\partial}_{I_1}(\n\cdot\omega)$ is bounded by a constant, and we apply \eqref{est_Omega1-2} to $\hat{\partial}_{I_2}(1/\sqrt{s})$. If $|I_1|=0$, hence $|I_2|\geq 1$, then the quantity $\hat{\partial}_{I_1}(\n\cdot\omega)$ is bounded by $|\n|$, and we apply \eqref{est_Omega1-3} to $\hat{\partial}_{I_2}(1/\sqrt{s})$. We obtain an estimate:
\[
\bigg|\hat{\partial}_I\bigg[\frac{(\hat{n}\cdot\omega)}{\sqrt{s}}\bigg]\bigg|\leq C\frac{(q^0)^{|I|}}{\sqrt{s}} \bigg(1+\frac{|\n|}{s}+\frac{|\n|}{p^0}\bigg),
\]
which is further estimated as follows: If $|\p|\leq\max\{2|\q|,1\}$, then $|\n|\leq Cq^0$, and we obtain $1+|\n|/s+|\n|/p^0\leq Cq^0$. If $|\p|\geq\max\{2|\q|,1\}$, then $|\n|\leq Cp^0$, and we apply \eqref{est_1/s} to obtain the same estimate, i.e., $1+|\n|/s+|\n|/p^0\leq Cq^0$. Therefore, we conclude that for $|I|\geq 1$,
\begin{align}\label{est_Omega2}
\bigg|\hat{\partial}_I\bigg[\frac{(\hat{n}\cdot\omega)}{\sqrt{s}}\bigg]\bigg|\leq C\frac{(q^0)^{|I|+1}}{\sqrt{s}}\leq \left\{
\begin{aligned}
C(q^0)^{|I|+1},\quad\mbox{if}\quad |\p|\leq\max\{2|\q|,1\},\\
C(p^0)^{-\frac12}(q^0)^{|I|+\frac32},\quad\mbox{if}\quad |\p|\geq\max\{2|\q|,1\},
\end{aligned}
\right.
\end{align}
where we used \eqref{est_1/s}.

The quantity $\hat{\partial}_I(\n/(n^0+\sqrt{s}))$ is estimated in a similar way. It is easy to see that for $|I|=0$,
\begin{align}\label{est_Omega3}
\bigg|\hat{\partial}_I\bigg[\frac{\n}{n^0+\sqrt{s}}\bigg]\bigg| \leq \frac{n^0}{n^0+\sqrt{s}}\leq C.
\end{align}
For $|I|\geq 1$, we write $\hat{\partial}_I(\n/(n^0+\sqrt{s}))$ as a finite sum of $(\hat{\partial}_{I_1}\n) \hat{\partial}_{I_2}(1/(n^0+\sqrt{s}))$ for some multi-indices $I_1$ and $I_2$ satisfying $I_1+I_2=I$.
If $|I_2|=0$, then $|I_1|\geq 1$, hence $|\hat{\partial}_{I_1}\n|\leq C$, and in this case we have
\[
\bigg|(\hat{\partial}_{I_1}\n) \hat{\partial}_{I_2}\bigg[\frac{1}{n^0+\sqrt{s}}\bigg]\bigg|\leq \frac{C}{n^0+\sqrt{s}}.
\]
If $|I_2|\geq 1$, then we apply the representation \eqref{partial_1/n0+s} to obtain
\begin{align*}
\bigg|(\hat{\partial}_{I_1}\n)\hat{\partial}_{I_2}\bigg[\frac{1}{n^0+\sqrt{s}}\bigg]\bigg| \leq \frac{Cn^0(q^0)^{|I_2|}}{(n^0+\sqrt{s})^2}\leq \frac{C(q^0)^{|I|}}{n^0+\sqrt{s}},
\end{align*}
where we used the fact that $|\hat{\partial}_{I_1}\n|\leq Cn^0$ for any $|I_1|\geq 0$. We combine these estimates to conclude that for $|I|\geq 1$,
\begin{align}\label{est_Omega4}
\bigg|\hat{\partial}_I\bigg[\frac{\n}{n^0+\sqrt{s}}\bigg]\bigg| \leq \frac{C(q^0)^{|I|}}{n^0+\sqrt{s}}\leq C(p^0)^{-1}(q^0)^{|I|},
\end{align}
where we simply used $n^0\geq p^0$.

We now collect the estimates \eqref{est_Omega1}, \eqref{est_Omega2}, \eqref{est_Omega3}, and \eqref{est_Omega4} to obtain
\begin{align}\label{est_Omega}
|\hat{\partial}_I\Omega^\alpha|&\leq \left\{
\begin{aligned}
C(p^0)^{\frac12}q^0,&\quad\mbox{if}\quad |I|=0,\\
C(q^0)^{|I|+1},&\quad\mbox{if}\quad |I|\geq 1,\quad |\p|\leq\max\{2|\q|,1\},\\
C(p^0)^{-\frac12}(q^0)^{|I|+\frac32},&\quad\mbox{if}\quad |I|\geq 1,\quad|\p|\geq\max\{2|\q|,1\}.
\end{aligned}
\right.
\end{align}
Hence, we observe from the representation of $\p'$ that high order derivatives of $\p'$ are estimated as follows:
\begin{align*}
|\hat{\partial}_I\p'|&\leq \left\{
\begin{aligned}
Cp^0(q^0)^3,&\quad\mbox{if}\quad |I|=0,\\
C(p^0)^{\frac12}(q^0)^{|I|+3},&\quad\mbox{if}\quad |I|\geq 1,\quad |\p|\leq\max\{2|\q|,1\},\\
C(q^0)^{|I|+4},&\quad\mbox{if}\quad |I|\geq 1,\quad|\p|\geq\max\{2|\q|,1\}.
\end{aligned}
\right.
\end{align*}
Since $p^0\leq Cq^0$ in the second case, we conclude that for $|I|\geq 1$,
\begin{align}\label{est_p'}
|\hat{\partial}_I\p'|\leq C(q^0)^{|I|+4},
\end{align}
and consequently the relation $\partial_Ip'_*=R^{1-|I|}\hat{\partial}_I\p'$ derives the desired result. The estimate of $\partial_Iq_*'$ is the same, and this completes the proof.
\end{proof}

\section{Energy estimate}\label{sec_energy}
In this part we study the energy estimate for the Boltzmann equation. We write the Boltzmann equation as follows:
\begin{align*}
\partial_tf&=Q(f,f)=Q_+(f,f)-Q_-(f,f),
\end{align*}
where the gain term $Q_+$ and the loss term $Q_-$ are written by
\begin{align*}
Q_+(f,f)&=R^{-3}\iint \frac{1}{p^0q^0\sqrt{s}}f(p_*')f(q_*')d\omega dq_*,\\
Q_-(f,f)&=R^{-3}\iint \frac{1}{p^0q^0\sqrt{s}}f(p_*)f(q_*)d\omega dq_*.
\end{align*}
Define the norm as follows: for $k\geq 0$ and $N\geq 0$,
\begin{align*}
\|g(t)\|^2_{k,N}=\sum_{|\beta|\leq N}\|\partial_\beta g(t)\|_k^2,\quad \|g(t)\|_k^2=\int_{\bbr^3}\langle p_*\rangle^{2k}e^{p^0(t)}|g(t,p_*)|^2dp_*,
\end{align*}
where $N$ and $k$ will be determined later.

\begin{lemma}\label{lem_f1}
Let $f$ be a solution of the Boltzmann equation. Then, $f$ satisfies the following estimate:
\[
\|f(t)\|_k^2\leq \|f(0)\|_k^2+C\sup_{s\in[0,t]}\|f(s)\|_k^3,
\]
where $k$ is a non-negative integer.
\end{lemma}
\begin{proof}
Multiplying $f$ to the Boltzmann equation, $\partial_t(f^2)/2=fQ(f,f)$, and integrating the equation on $[0,t]$ at a fixed $p_*$, we have
\[
f^2(t,p_*)=f^2(0,p_*)+2\int_0^tf(s,p_*)Q(f,f)(s,p_*)ds.
\]
We consider $p^0=p^0(t)$ as a function of $t$, and multiply $e^{p^0}$ to the above to obtain
\begin{align*}
e^{p^0(t)}f^2(t,p_*)&=e^{p^0(t)}f^2(0,p_*)+2\int_0^te^{p^0(t)}f(s,p_*)Q(f,f)(s,p_*)ds\\
&=e^{p^0(t)}f^2(0,p_*)+2\int_0^tR^{-3}(s)f(s,p_*)\iint \frac{1}{p^0q^0\sqrt{s}}e^{p^0(t)}f(s,p_*')f(s,q_*')d\omega dq_* ds\\
&\quad -2\int_0^t R^{-3}(s)f(s,p_*)\iint \frac{1}{p^0q^0\sqrt{s}}e^{p^0(t)}f(s,p_*)f(s,q_*)d\omega dq_* ds.
\end{align*}
Note that since $R$ is increasing, $p^0$ is decreasing in $t$ for each $p_*$, hence we have
\begin{align*}
&e^{p^0(t)}f^2(t,p_*)\leq e^{p^0(0)}f^2(0,p_*)\\
&\quad+2\int_0^tR^{-3}(s)f(s,p_*)\iint \frac{1}{p^0q^0\sqrt{s}}e^{p^0(s)}f(s,p_*')f(s,q_*')d\omega dq_* ds,
\end{align*}
where we ignored the loss term. By the energy conservation at time $s$, i.e.,
\[
p'^0(s)+q'^0(s)=p^0(s)+q^0(s),
\]
we have
\begin{align*}
&e^{p^0(t)}f^2(t,p_*)\leq e^{p^0(0)}f^2(0,p_*)\\
&\quad+2\int_0^tR^{-3}(s)e^{\frac12p^0(s)}f(s,p_*)\iint \frac{1}{p^0q^0\sqrt{s}}e^{-\frac12q^0(s)}e^{\frac12p'^0(s)}f(s,p_*')e^{\frac12q'^0(s)}f(s,q_*')d\omega dq_* ds.
\end{align*}
We now multiply the weight function $\langle p_*\rangle$ to the above. By Lemma \ref{lem_pp'q'} with $k\geq 0$ we have
\begin{align*}
&\langle p_*\rangle^{2k} e^{p^0(t)}f^2(t,p_*)\leq\langle p_*\rangle^{2k} e^{p^0(0)}f^2(0,p_*)\\
&\quad +C\int_0^tR^{-3}(s)\langle p_*\rangle^ke^{\frac12p^0(s)}f(s,p_*)\\
&\qquad\times\iint \frac{1}{p^0q^0\sqrt{s}}e^{-\frac12q^0(s)}\langle p_*'\rangle^ke^{\frac12p'^0(s)}f(s,p_*')\langle q_*'\rangle^ke^{\frac12q'^0(s)}f(s,q_*')d\omega dq_* ds.
\end{align*}
Integrate the above inequality with respect to $p_*$ to obtain $\|f(t)\|_k^2$ and $\|f(0)\|_k^2$ from the left side and the first quantity on the right side, respectively. To estimate the last quantity we write $F(s,p_*)=\langle p_*\rangle^k e^{\frac12 p^0(s)}f(s,p_*)$ for simplicity, and consider for each $s$,
\begin{align*}
&\int F(p_*)\iint\frac{1}{p^0q^0\sqrt{s}}e^{-\frac12 q^0(s)}F(p_*')F(q_*')d\omega dq_*dp_*\\
&\leq C\bigg(\iiint F^2(p_*)e^{-q^0(s)}d\omega dq_* dp_*\bigg)^{\frac12} \bigg(\iiint \frac{1}{p^0q^0}F^2(p_*')F^2(q_*')d\omega dq_*dp_*\bigg)^{\frac12}\allowdisplaybreaks\\
&\leq C\|F(s)\|_{L^2}\bigg(\int e^{-q^0(s)} dq_*\bigg)^{\frac12} \bigg(\iiint \frac{1}{p^0q^0}F^2(p_*)F^2(q_*)d\omega dq_*dp_*\bigg)^{\frac12}\\
&\leq CR^{\frac32}(s)\|F(s)\|_{L^2}^3,
\end{align*}
where we used $(p^0q^0)^{-1}dp_*dq_*=(p'^0q'^0)^{-1}dp_*'dq_*'$ and Lemma \ref{lem_R}. To summarize, we have
\[
\|f(t)\|_k^2\leq \|f(0)\|_k^2+C\int_0^tR^{-\frac32}(s)\|F(s)\|_{L^2}^3ds,
\]
and, since $R^{-\frac32}$ is integrable and $\|F(s)\|_{L^2}=\|f(s)\|_k$, we obtain the desired result.
\end{proof}

\begin{lemma}\label{lem_f2}
Let $f$ be a solution of the Boltzmann equation. Then, $f$ satisfies the following estimate: for $\beta\neq 0$,
\[
\|\partial_{\beta}f(t)\|^2_k\leq \|\partial_\beta f(0)\|^2_{k}+C\sup_{s\in[0,t]}\|f(s)\|_{k,|\beta|}^3,
\]
where $k$ is a non-negative integer.
\end{lemma}
\begin{proof}
For a multi-index $\beta\neq 0$, we take $\partial_\beta$ to the Boltzmann equation,
\begin{align*}
\partial_t\partial_\beta f&=R^{-3}\sum\iint \partial_{\beta_0}\bigg[\frac{1}{p^0q^0\sqrt{s}}\bigg] \partial_{\beta_1}\Big[f(p_*')\Big]\partial_{\beta_2}\Big[f(q_*')\Big]d\omega dq_*\\
&\quad -R^{-3}\sum\iint \partial_{\beta_0}\bigg[\frac{1}{p^0q^0\sqrt{s}}\bigg] \partial_{\beta_1}f(p_*)f(q_*)d\omega dq_*,
\end{align*}
where the summations are finite and taken for some $\beta_0$, $\beta_1$, and $\beta_2$ satisfying $\beta_0+\beta_1+\beta_2=\beta$ or $\beta_0+\beta_1=\beta$. Multiply $2\partial_\beta f$ and take integration on $[0,t]$ to obtain
\begin{align*}
&(\partial_\beta f)^2(t,p_*)=(\partial_\beta f)^2(0,p_*)\\
&\quad +2\int_0^tR^{-3}(s)\partial_\beta f(s,p_*)\sum\iint \partial_{\beta_0}\bigg[\frac{1}{p^0q^0\sqrt{s}}\bigg] \partial_{\beta_1}\Big[f(p_*')\Big]\partial_{\beta_2}\Big[f(q_*')\Big]d\omega dq_*ds\\
&\quad -2\int_0^tR^{-3}(s)\partial_\beta f(s,p_*)\sum\iint \partial_{\beta_0}\bigg[\frac{1}{p^0q^0\sqrt{s}}\bigg] \partial_{\beta_1}f(p_*)f(q_*)d\omega dq_*ds.
\end{align*}
As in the previous lemma, we multiply $\langle p_*\rangle^{2k} e^{p^0(t)}$ to the above equation, and use the decreasing property of $p^0(t)$ and Lemma \ref{lem_pp'q'} as follows:
\begin{align}
&\langle p_*\rangle^{2k} e^{p^0(t)}(\partial_\beta f)^2(t,p_*)\leq \langle p_*\rangle^{2k} e^{p^0(0)}(\partial_\beta f)^2(0,p_*)\nonumber\\
&\quad +C\sum\int_0^tR^{-3}(s)\langle p_*\rangle^{k} e^{\frac12 p^0(s)}|\partial_\beta f(s,p_*)|\nonumber\\
&\qquad\times\iint \bigg|\partial_{\beta_0}\bigg[\frac{1}{p^0q^0\sqrt{s}}\bigg]\bigg| e^{-\frac12 q^0(s)}\langle p_*'\rangle^{k} e^{\frac12 p'^0(s)}\Big|\partial_{\beta_1}\Big[f(p_*')\Big]\Big|\langle q_*'\rangle^{k} e^{\frac12 q'^0(s)}\Big|\partial_{\beta_2}\Big[f(q_*')\Big]\Big|d\omega dq_*ds\allowdisplaybreaks\nonumber\\
&\quad +C\sum\int_0^tR^{-3}(s)\langle p_*\rangle^{k} e^{\frac12 p^0(s)}|\partial_\beta f(s,p_*)|\nonumber\\
&\qquad\times\iint \bigg|\partial_{\beta_0}\bigg[\frac{1}{p^0q^0\sqrt{s}}\bigg]\bigg|e^{-\frac12 q^0(s)}\langle p_*\rangle^{k} e^{\frac12 p^0(s)}| \partial_{\beta_1}f(p_*)|\langle q_*\rangle^{k} e^{\frac12 q^0(s)}f(q_*)d\omega dq_*ds.\label{est_partial_f1}
\end{align}
To estimate the right hand side of \eqref{est_partial_f1}, we first consider the partial derivatives. Since $p_*/(Rp^0)$, $q_*/(Rq^0)$, $1/s$, and $1/p^0$ are bounded quantities, we obtain by Lemma \ref{lem_partial_1/p0} and \ref{lem_partial_1/s} for a multi-index $I$,
\begin{align}\label{est_partial_f2}
\bigg|\partial_{I}\bigg[\frac{1}{p^0q^0\sqrt{s}}\bigg]\bigg|\leq \sum_{I_1+I_2=I}\frac{C(q^0)^{|I|}}{R^{|I_{1}|}(p^0)^{|I_1|+1} q^0 R^{|I_2|}\sqrt{s}}\leq \frac{C(q^0)^{|I|}}{p^0q^0},
\end{align}
where we simply ignored $R$ and $\sqrt{s}$. To estimate $\partial_{\beta_1}[f(p_*')]$, we apply Faa di Bruno's formula as in \cite{ConstantineSavits96,GuoStrain12}. Applying the main theorem of \cite{ConstantineSavits96} to our case, we obtain for a multi-index $I\neq 0$,
\[
\partial_I\Big[f(p_*')\Big]=\sum (\partial_{I_1}f)(p_*')[\partial_{J_1}p_*']^{K_1}\cdots[\partial_{J_s}p_*']^{K_s},
\]
which is a finite sum for multi-indices $I_1$, $J_i$, and $K_i$ such that $1\leq|I_1|\leq |I|$, $|J_i|\geq 1$, $|K_i|\geq 1$, and $|I|=\sum_{i=1}^s|J_i||K_i|$, and $[\partial_{J_i}p_*']^{K_i}$ is understood as $(\partial_{J_i}p_1')^{k_{i1}}(\partial_{J_i}p_2')^{k_{i2}}(\partial_{J_i}p_3')^{k_{i3}}$ for $K_i=(k_{i1},k_{i2},k_{i3})$. Applying Lemma \ref{lem_p'}, we have
\[
\Big|[\partial_{J_1}p_*']^{K_1}\cdots[\partial_{J_s}p_*']^{K_s}\Big|\leq C\prod_{i=1}^s\Big(R^{1-|J_i|}(q^0)^{|J_i|+4}\Big)^{|K_i|}\leq C(q^0)^{5|I|},
\]
where we used $|J_i|\geq 1$ and $|I|=\sum_{i=1}^s|J_i||K_i|$. Hence, we obtain
\begin{align}\label{est_partial_f3}
\Big|\partial_I\Big[f(p_*')\Big]\Big|\leq C(q^0)^{5|I|}\sum |(\partial_{I_1}f)(p_*')|,
\end{align}
where the summation is over $1\leq|I_1|\leq |I|$, and the same estimate for $\partial_I[f(q_*')]$. Note that the estimate \eqref{est_partial_f3} still holds for $I=0$, in which case we understand $I_1=0$.

We are now ready to estimate \eqref{est_partial_f1}. Applying \eqref{est_partial_f2} and \eqref{est_partial_f3} to the integrands in the second and the third terms of \eqref{est_partial_f1}, we obtain
\begin{align}\label{est_partial_f4}
&\langle p_*\rangle^{2k} e^{p^0(t)}(\partial_\beta f)^2(t,p_*)\leq \langle p_*\rangle^{2k} e^{p^0(0)}(\partial_\beta f)^2(0,p_*)\nonumber\\
&\quad +C\sum\int_0^tR^{-3}(s)\langle p_*\rangle^{k} e^{\frac12 p^0(s)}|\partial_\beta f(s,p_*)|\nonumber\\
&\qquad\times\iint \frac{(q^0)^{5|\beta|}}{p^0q^0} e^{-\frac12 q^0(s)}\langle p_*'\rangle^{k} e^{\frac12 p'^0(s)}|(\partial_{I_1}f)(p_*')|\langle q_*'\rangle^{k} e^{\frac12 q'^0(s)}|(\partial_{I_2}f)(q_*')|d\omega dq_*ds\allowdisplaybreaks\nonumber\\
&\quad +C\sum\int_0^tR^{-3}(s)\langle p_*\rangle^{k} e^{\frac12 p^0(s)}|\partial_\beta f(s,p_*)|\nonumber\\
&\qquad\times\iint \frac{(q^0)^{|\beta_0|}}{p^0q^0}e^{-\frac12 q^0(s)}\langle p_*\rangle^{k} e^{\frac12 p^0(s)}| \partial_{\beta_1}f(p_*)|\langle q_*\rangle^{k} e^{\frac12 q^0(s)}f(q_*)d\omega dq_*ds,
\end{align}
where the summation of the second term is over some $I_1$ and $I_2$ satisfying $|I_1|+|I_2|\leq |\beta|$. The estimate is now almost the same with that of Lemma \ref{lem_f1}. Integrating the above inequality with respect to $p_*$, we obtain $\|\partial_\beta f(t)\|_k^2$ and $\|\partial_\beta f(0)\|_k^2$ from the left side and the first term on the right side, respectively. To estimate the second term, we write $F_{\beta}(s,p_*)=\langle p_*\rangle^{k} e^{\frac12 p^0(s)}|\partial_\beta f(s,p_*)|$, and consider for each $s$,
\begin{align*}
&\int F_\beta(p_*)\iint \frac{(q^0)^{5|\beta|}}{p^0q^0} e^{-\frac12 q^0(s)}F_{I_1}(p_*')F_{I_2}(q_*')d\omega dq_*dp_*\\
&\leq C\bigg(\iiint F^2_\beta(p_*)(q^0)^{10|\beta|}e^{-q^0(s)}d\omega dq_* dp_*\bigg)^{\frac12} \bigg(\iiint\frac{1}{p^0q^0}F^2_{I_1}(p_*')F^2_{I_2}(q_*')d\omega dq_*dp_*\bigg)^{\frac12}\allowdisplaybreaks\\
&\leq C\|F_\beta(s)\|_{L^2}\bigg(\int (q^0)^{10|\beta|}e^{-q^0(s)}dq_*\bigg)^{\frac12} \bigg(\iiint\frac{1}{p^0q^0}F^2_{I_1}(p_*)F^2_{I_2}(q_*)d\omega dq_*dp_*\bigg)^{\frac12}\\
&\leq CR^{\frac32}(s)\|F_\beta(s)\|_{L^2}\|F_{I_1}(s)\|_{L^2}\|F_{I_2}(s)\|_{L^2},
\end{align*}
where we used $(p^0q^0)^{-1}dp_*dq_*=(p'^0q'^0)^{-1}dp_*'dq_*'$ and Lemma \ref{lem_R}. This shows that the second term of \eqref{est_partial_f4} is estimated as
\begin{align*}
\sum\int_0^tR^{-\frac32}(s)\|F_\beta(s)\|_{L^2}\|F_{I_1}(s)\|_{L^2}\|F_{I_2}(s)\|_{L^2}ds
\leq C\sup_{s\in[0,t]}\bigg(\sum_{|\beta_1|\leq|\beta|}\|F_{\beta_1}(s)\|_{L^2}\bigg)^3,
\end{align*}
since $R^{-\frac32}$ is integrable. In a similar way we estimate the loss term and obtain the same upper bound as in the gain term. We skip the estimate of the loss term, and conclude that $f$ satisfies
\[
\|\partial_\beta f(t)\|_k^2\leq \|\partial_\beta f(0)\|_k^2+C\sup_{s\in[0,t]}\|f(s)\|^3_{k,|\beta|},
\]
and this completes the proof.
\end{proof}

\begin{remark}
Lemma \ref{lem_f1} and \ref{lem_f2} show that if $f$ is a local-in-time solution of the Boltzmann equation, then it satisfies for any $N\geq 0$ and $k\geq 0$,
\begin{equation}
\|f(t)\|^2_{k,N}\leq \|f(0)\|^2_{k,N}+C\sup_{s\in[0,t]}\|f(s)\|_{k,N}^3,\label{remark_1}
\end{equation}
on a (short) time interval. This proves that the solution is extended to a global-in-time solution, if initial data is given such that $\|f(0)\|^2_{k,N}$ is sufficiently small.
\end{remark}

\section{Main result}\label{sec_main}
\subsection{Global-in-time existence for the Boltzmann equation}
Local-in-time existence is proved by a standard iteration method. Let us consider the following iteration:
\begin{align*}
\partial_tf^{n+1}=R^{-3}\int_{\bbr^3}\int_{\bbs^2} v_M\sigma(h,\theta)\Big(f^n(p_*')f^n(q_*')-f^{n+1}(p_*)f^n(q_*)\Big)d\omega dq_*,
\end{align*}
with $f^{n+1}(0,p_*)=f(0,p_*)$ and $f^0(t,p_*)=f(0,p_*)$, and obtain a sequence $\{f^n\}_{n=0}^\infty$. As in Lemma \ref{lem_f1} and \ref{lem_f2}, we take a partial derivative $\partial_\beta$ for $\beta\geq 0$ and multiply $2\partial_\beta f^{n+1}$ to the above equation. To the equation obtained, we first take integration on $[0,t]$, and then multiply $\langle p_*\rangle^{2k}e^{p^0(t)}$ and use the decreasing property of $p^0(t)$ to obtain the following inequality:
\begin{align*}
&\langle p_*\rangle^{2k} e^{p^0(t)}(\partial_\beta f^{n+1})^2(t,p_*)\leq \langle p_*\rangle^{2k} e^{p^0(0)}(\partial_\beta f)^2(0,p_*)\nonumber\\
&\quad +C\sum\int_0^tR^{-3}(s)\langle p_*\rangle^{k} e^{\frac12 p^0(s)}|\partial_\beta f^{n+1}(s,p_*)|\nonumber\\
&\qquad\times\iint \bigg|\partial_{\beta_0}\bigg[\frac{1}{p^0q^0\sqrt{s}}\bigg]\bigg| e^{-\frac12 q^0(s)}\langle p_*'\rangle^{k} e^{\frac12 p'^0(s)}\Big|\partial_{\beta_1}\Big[f^n(p_*')\Big]\Big|\langle q_*'\rangle^{k} e^{\frac12 q'^0(s)}\Big|\partial_{\beta_2}\Big[f^n(q_*')\Big]\Big|d\omega dq_*ds\allowdisplaybreaks\nonumber\\
&\quad +C\sum\int_0^tR^{-3}(s)\langle p_*\rangle^{k} e^{\frac12 p^0(s)}|\partial_\beta f^{n+1}(s,p_*)|\nonumber\\
&\qquad\times\iint \bigg|\partial_{\beta_0}\bigg[\frac{1}{p^0q^0\sqrt{s}}\bigg]\bigg|e^{-\frac12 q^0(s)}\langle p_*\rangle^{k} e^{\frac12 p^0(s)}| \partial_{\beta_1}f^{n+1}(p_*)|\langle q_*\rangle^{k} e^{\frac12 q^0(s)}f^n(q_*)d\omega dq_*ds.
\end{align*}
Integrate the above inequality to obtain $\|\partial_\beta f^{n+1}(t)\|_k^2$ and $\|\partial_\beta f(0)\|_k^2$ from the left side and the first quantity on the right side, respectively. Following the proofs of Lemma \ref{lem_f1} and \ref{lem_f2} and summing over $|\beta|\leq N$, we obtain the following estimate:
\begin{align*}
&\|f^{n+1}(t)\|_{k,N}^2\\ &\quad \leq \|f(0)\|_{k,N}^2
 +C\sup_{s\in[0,t]}\|f^{n+1}(s)\|_{k,N} \|f^{n}(s)\|^2_{k,N} +C\sup_{s\in[0,t]}\|f^{n+1}(s)\|^2_{k,N} \|f^{n}(s)\|_{k,N}\\
&\quad \leq \|f(0)\|_{k,N}^2+C\sup_{s\in[0,t]} \|f^{n}(s)\|^3_{k,N} +C\sup_{s\in[0,t]}\|f^{n+1}(s)\|^2_{k,N} \|f^{n}(s)\|_{k,N},
\end{align*}
where we used a simple inequality $ab^2\leq a^2b+b^3$ for positive $a$ and $b$. Suppose now that there exists a positive $M$ such that $\|f(0)\|^2_{k,N}\leq M/2$ and $\|f^{n}(t)\|^2_{k,N}\leq M$ on a time interval $[0,T]$. Then, we have
\begin{align*}
&\|f^{n+1}(t)\|_{k,N}^2\leq M/2 +CM^{\frac32} +C\sqrt{M}\sup_{s\in[0,t]}\|f^{n+1}(s)\|^2_{k,N},
\end{align*}
which shows that if $M$ is sufficiently small, then $\|f^{n+1}(t)\|^2_{k,N}\leq M$ on $[0,T]$, and we conclude that $f^n$ is bounded uniformly on $n$. Taking limit $n\to\infty$, we have a local-in-time solution such that $\|f(t)\|^2_{k,N}\leq M$ on $[0,T]$, and the inequality \eqref{remark_1} now proves that if $\|f(0)\|^2_{k,N}$ is sufficiently small, then the solution exists globally in time. Non-negativity of solutions is guaranteed by the iteration, and uniqueness is easily proved as in \cite{Guo03}. For more details, we refer to \cite{Guo03,GuoStrain12}.

\begin{thm}
Consider a spatially flat FLRW spacetime where the scale factor $R=R(t)$ is given by an increasing function with an exponential rate. Let $f(0)=f(0,p_*)$ be an initial data of the Boltzmann equation \eqref{Boltzmann} such that $\|f(0)\|^2_{k,N}$ is bounded with $k\geq 0$ and $N\geq 3$. Then, there exists an $\varepsilon>0$ such that if $\|f(0)\|^2_{k,N}<\varepsilon$, then the corresponding solution exists globally in time such that
\[
\sup_{0\leq t<\infty}\|f(t)\|^2_{k,N}\leq C\varepsilon,
\]
for some constant $C>0$.
\end{thm}

\begin{remark}
The Boltzmann equation is written as in an orthonormal frame,
\begin{align}
\partial_t\hat{f}-(\dot{R}\hat{p}/R)\cdot\nabla_{\hat{p}}\hat{f}=Q(\hat{f},\hat{f}),\label{Boltmzann_ortho}
\end{align}
where $\hat{f}=\hat{f}(t,\p)$ denotes that the distribution function is written in an orthonormal frame to have variables $t$ and $\p=R^{-1}p_*$. Let $f$ be a solution of the Boltzmann equation \eqref{Boltzmann} constructed in the previous theorem. Since $N\geq 3$, $f$ is of $C^1$ in $p_*$, hence the equations \eqref{Boltzmann} and \eqref{Boltmzann_ortho} are equivalent, and $\hat{f}(t,\p)=f(t,p_*)$ solves the equation \eqref{Boltmzann_ortho}. Since $f(t,p_*)\leq C\varepsilon \langle p_*\rangle^{-k}e^{-\frac{1}{2}p^0(t)}$, we now have the asymptotic behaviour as follows:
\[
\hat{f}(t,\p)\leq C\varepsilon(1+R^2(t)|\p|^2)^{-\frac{k}{2}}e^{-\frac{1}{2}p^0},
\]
where $p^0=\sqrt{1+|\p|^2}$ is independent of $t$ in an orthonormal frame.
\end{remark}

\section{Summary and outlook}\label{sec_summary}
The result of an accelerated expansion has been obtained for a given distribution function in Section \ref{sec_flrw}. The existence of classical solutions to the Boltzmann equation has been shown in Section \ref{sec_main} and presents the core of this paper. One may now consider the coupled Einstein-Boltzmann system, where the equations \eqref{Friedmann}--\eqref{f2} are coupled to the Boltzmann equation \eqref{Boltzmann} through the energy density and the pressure, which are given by
\begin{align*}
\rho &=  R^{-3} \int_{\bbr^3} f(t,p_*)p^0dp_*,\\
P &= R^{-5} \int_{\bbr^3} f(t,p_*) \frac{|p_*|^2}{3p^0}dp_*,
\end{align*}
where the distribution function again is understood as a function of time $t$ and covariant momenta $p_*=(p_1,p_2,p_3)$. Since we have considered particles of unit rest mass, which are future pointing, i.e., $p_\alpha p_\beta g^{\alpha\beta}= -1$ and $p^0 >0$, and the distribution function is non-negative, it is easy to see that the matter, described by the Boltzmann equation, satisfies the weak and the dominant energy conditions. In this sense, Lemma \ref{lem_flrw} applies to the coupled system as well, and in the present paper we only remark that a standard iteration method will give us the same results.

Using a new representation for the collision operator and the post-collision momentum we have obtained future global existence of classical solutions for Israel particles and shown that they have an asymptotic behaviour in a spatially flat FLRW spacetime with a positive cosmological constant. There are not so many results concerning solutions to the Einstein-Boltzmann system, and moreover most of them impose artificial restrictions on the scattering kernel for technical reasons. We have been able to obtain our results without any artificial restrictions. The results have been achieved however only for small solutions in the presence of a positive cosmological constant, which has been used in a crucial way. In any case the results obtained open the possibility to explore in more detail scattering kernels which are physically relevant. It is clear that the Universe is not exactly isotropic, so it will be of interest to consider more general models such as the homogeneous, but anisotropic Bianchi I spacetime. Further it should be possible to consider forever expanding Bianchi spacetimes as was done for the Vlasov case in \cite{Lee04}. Finally, we remark that on the basis of the monograph \cite{Hans} (cf.\ in particular p.\ 154--156 and Appendix F) it is of interest to consider stability in the general case.

\section*{Appendix}
In this part we derive the representation of the collision operator that we used in this paper. The main idea is to use a Lorentz transform as in \cite{Strain11}, and calculation is slightly modified. Let us consider the Minkowski case, i.e., $g_{\alpha\beta}=\eta_{\alpha\beta}$, and suppose that $p^\alpha$ and $q^\alpha$ are pre-collision momenta satisfying the mass shell condition. Let $\Lambda$ be a Lorentz transform satisfying
\[
\Lambda^\alpha_\beta n^\beta=(\sqrt{s},0,0,0),\quad \Lambda^\alpha_\beta \Lambda^\gamma_\delta \eta_{\alpha\gamma}=\eta_{\beta\delta},
\]
where $n^\alpha=p^\alpha+q^\alpha$. For simplicity, we use tildes to denote transformed quantities, for  instance $\tilde{p}^\alpha=\Lambda^\alpha_\beta p^\beta$, $\tilde{p}'^\alpha=\Lambda^\alpha_\beta p'^\beta$, and so on. Since $\Lambda$ is a Lorentz transform, the mass shell condition still holds for pre- and post-collision momenta, i.e., $\tilde{p}_\alpha\tilde{p}^\alpha=\tilde{q}_\alpha\tilde{q}^\alpha=\tilde{p}'_\alpha\tilde{p}'^\alpha=\tilde{q}'_\alpha\tilde{q}'^\alpha=-1$, which derive
\[
\tilde{p}^0=\sqrt{1+|\tilde{p}|^2},\quad\tilde{p}'^0=\sqrt{1+|\tilde{p}'|^2},
\]
and similar representations for $\tilde{q}^0$ and $\tilde{q}'^0$. Since $\tilde{n}^0=\sqrt{s}$ and $\tilde{n}^k=0$, we have $\tilde{p}^0=\tilde{q}^0=\sqrt{s}/2$ and $|\tilde{p}|=|\tilde{q}|=h/2$. The collision operator is written as
\[
Q(f,f)=\iiint W(p,q,p',q')G(p,q,p',q')\frac{dp'dq'dq}{p'^0q'^0p^0q^0},
\]
where $G$ denotes $f(p')f(q')-f(p)f(q)$. The transition rate $W$ is given by
\[
W(p,q,p',q')=\frac{s}{2}\sigma(h,\theta)\delta^{(4)}(p^\alpha+q^\alpha-p'^\alpha-q'^\alpha),
\]
where $\delta^{(4)}$ is the four-dimensional Dirac delta function. For given $p$ and $q$, let us consider
\begin{align*}
I&=\iint W(p,q,p',q')G(p,q,p',q')\frac{dp'dq'}{p'^0q'^0}\\
&=\iint\frac{s\sigma}{2}\delta^{(4)}(p^\alpha+q^\alpha-p'^\alpha-q'^\alpha)G(\cdots)\frac{dp'dq'}{p'^0q'^0}.
\end{align*}
Since $(p'^0)^{-1}dp'$ and $(q'^0)^{-1}dq'$ are Lorentz invariant, we have
\begin{align*}
I&=\iint\frac{s\sigma}{2}\delta^{(4)}(\tilde{p}^\alpha+\tilde{q}^\alpha-\tilde{p}'^\alpha-\tilde{q}'^\alpha)G(\cdots)\frac{d\tilde{p}'d\tilde{q}'}{\tilde{p}'^0\tilde{q}'^0}.
\end{align*}
Since $\tilde{p}^k+\tilde{q}^k=0$, the quantity $\delta^{(3)}(\tilde{p}^k+\tilde{q}^k-\tilde{p}'^k-\tilde{q}'^k)$ derives $\tilde{q}'^k=-\tilde{p}'^k$, and the mass shell condition shows that $\tilde{p}'^0=\tilde{q}'^0$. Since we already have $\tilde{p}^0=\tilde{q}^0$, the quantity $I$ reduces to
\begin{align*}
I&=\int\frac{s\sigma}{2}\delta(2\tilde{p}^0-2\tilde{p}'^0)G(\cdots)\frac{d\tilde{p}'}{(\tilde{p}'^0)^2}
=\int\frac{s\sigma}{4}\delta(\tilde{p}'^0-\tilde{p}^0)G(\cdots)\frac{d\tilde{p}'}{(\tilde{p}'^0)^2}.
\end{align*}
To compute the quantity $\delta(\tilde{p}'^0-\tilde{p}^0)$, we take the translation $\bar{p}'^k=\tilde{p}'^k-\tilde{p}^k$ for $k=1,2,3$, and use the mass shell conditions to have
\begin{align*}
\tilde{p}'^0-\tilde{p}^0&=\frac{|\tilde{p}'|^2-|\tilde{p}|^2}{\tilde{p}'^0+\tilde{p}^0}=\frac{(\tilde{p}'+\tilde{p})\cdot(\tilde{p}'-\tilde{p})}{\tilde{p}'^0+\tilde{p}^0}=\frac{(\bar{p}'+2\tilde{p})\cdot\bar{p}'}{\tilde{p}'^0+\tilde{p}^0}.
\end{align*}
We write $\bar{p}'=|\bar{p}'|\omega$ for $\omega\in\bbs^2$ to obtain
\begin{align*}
\delta(\tilde{p}'^0-\tilde{p}^0)&=\delta\bigg(\frac{|\bar{p}'|(|\bar{p}'|+2\tilde{p}\cdot\omega)}{\tilde{p}'^0+\tilde{p}^0}\bigg)=\frac{(\tilde{p}'^0+\tilde{p}^0)}{|2\tilde{p}\cdot\omega|}\Big(\delta(|\bar{p}'|)+\delta(|\bar{p}'|+2\tilde{p}\cdot\omega)\Big),
\end{align*}
where we used Lemma 1.3.1 of \cite{Glassey}. Hence, we have
\begin{align*}
I&=\int\frac{s\sigma}{4}\frac{(\tilde{p}'^0+\tilde{p}^0)}{|2\tilde{p}\cdot\omega|}\Big(\delta(|\bar{p}'|)+\delta(|\bar{p}'|+2\tilde{p}\cdot\omega)\Big)G(\cdots)\frac{|\bar{p}'|^2d|\bar{p}'|d\omega}{(\tilde{p}'^0)^2}\\
&=\int\frac{s\sigma}{4}\frac{(\tilde{p}'^0+\tilde{p}^0)}{|2\tilde{p}\cdot\omega|}\delta(|\bar{p}'|+2\tilde{p}\cdot\omega)G(\cdots)\frac{|\bar{p}'|^2d|\bar{p}'|d\omega}{(\tilde{p}'^0)^2},
\end{align*}
where $\delta(|\bar{p}'|)$ vanishes by the quantity $|\bar{p}'|^2$. Note that $\tilde{p}'^0=\tilde{p}^0$, when $|\bar{p}'|=-2\tilde{p}\cdot\omega$, and therefore we have
\begin{align*}
I&=\int\frac{s\sigma}{4}\frac{2\tilde{p}^0}{2|\tilde{p}\cdot\omega|}G(\cdots)\frac{4|\tilde{p}\cdot\omega|^2d\omega}{(\tilde{p}^0)^2}=\int \frac{s\sigma|\tilde{p}\cdot\omega|}{\tilde{p}^0}G(\cdots)d\omega.
\end{align*}
Since $\tilde{p}=-\tilde{q}$ and $\tilde{p}^0=\sqrt{s}/2$, we obtain
\[
I=\int2\sqrt{s}|\tilde{q}\cdot\omega|\sigma(h,\theta)G(\cdots)d\omega.
\]
Note that $\tilde{q}$ is a transformed variable, and we consider the inverse of the Lorentz transform. By introducing a four-dimensional vector $\omega^\alpha=(0,\omega)$, we have the following representation:
\begin{align*}
\tilde{q}\cdot\omega=\eta_{\alpha\beta}\tilde{q}^\alpha\omega^\beta=\eta_{\alpha\beta}\Lambda^\alpha_\gamma q^\gamma\Lambda^\beta_\delta(\Lambda^{-1})^\delta_\mu\omega^\mu=\eta_{\gamma\delta}q^\gamma(\Lambda^{-1})^\delta_\mu\omega^\mu=q_\alpha\Omega^\alpha,
\end{align*}
where we defined $\Omega^\alpha=(\Lambda^{-1})^\alpha_\beta\omega^\beta$. The post-collision momentum is given by $\tilde{p}'^0=\tilde{p}^0$ and $\bar{p}'=|\bar{p}'|\omega$ for $\omega\in\bbs^2$. Since $\bar{p}'=\tilde{p}'-\tilde{p}$ and $|\bar{p}'|=-2\tilde{p}\cdot\omega=2\tilde{q}\cdot\omega$, we have $\tilde{p}'=\tilde{p}+2(\tilde{q}\cdot\omega)\omega$. Then, we can write $\tilde{p}'^\alpha=\tilde{p}^\alpha+2(\tilde{q}\cdot\omega)\omega^\alpha$. Taking the inverse of the Lorentz transform, we have
\begin{align}
p'^\alpha=p^\alpha+2(q_\beta\Omega^\beta)\Omega^\alpha,\quad q'^\alpha=q^\alpha-2(q_\beta\Omega^\beta)\Omega^\alpha.\label{p'_alpha}
\end{align}
Note also that $\Omega^\alpha$ is a unit spacelike vector and $n_\alpha\Omega^\alpha=\tilde{n}\cdot\omega=0$, which shows that it is orthogonal to $p^\alpha+q^\alpha$, hence we may write the post-collision momentum as $p'^\alpha=p^\alpha+((q_\beta-p_\beta)\Omega^\beta)\Omega^\alpha$. Then, we have $p'^\alpha-q'^\alpha=p^\alpha-q^\alpha+2((q_\beta-p_\beta)\Omega^\beta)\Omega^\alpha$ and multiply $p_\alpha-q_\alpha$ to obtain
\begin{align*}
h^2\cos\theta=(p_\alpha-q_\alpha)(p'^\alpha-q'^\alpha)=h^2-2((p_\beta-q_\beta)\Omega^\beta)^2,
\end{align*}
and therefore $((p_\alpha-q_\alpha)\Omega^\alpha)^2=h^2(1-\cos\theta)/2=h^2\sin^2(\theta/2)$. Since $2|\tilde{q}\cdot\omega|=|(p_\alpha-q_\alpha)\Omega^\alpha|$, we have the following representation for $I$:
\[
I=\int h\sqrt{s}\sin(\theta/2)\sigma(h,\theta)G(\cdots)d\omega.
\]
By abuse of notation we may redefine the scattering kernel as $4\sin(\theta/2)\sigma(h,\theta)$ and introduce the M{\o}ller velocity $v_M=h\sqrt{s}/(4p^0q^0)$ to obtain the following representation of $Q$:
\begin{align}
Q(f,f)=\int_{\bbr^3}\int_{\bbs^2} v_M\sigma(h,\theta)\Big(f(p')f(q')-f(p)f(q)\Big)d\omega dq.\label{Q_ortho}
\end{align}
For the Lorentz transform, we use the boost matrix as in \cite{Strain11}:
\begin{align*}
(\Lambda^{-1})^\alpha_\beta\omega^\beta=\Omega^\alpha=\bigg(\frac{(n\cdot\omega)}{\sqrt{s}},\omega+\frac{(n\cdot\omega)n}{\sqrt{s}(n^0+\sqrt{s})}\bigg).
\end{align*}
Then, the post-collision momentum is given by
\begin{align}
\left(
\begin{array}{c}
p'^0\\
p'^k
\end{array}
\right)=
\left(
\begin{array}{c}
\displaystyle
p^0+2\bigg(-q^0\frac{(n\cdot\omega)}{\sqrt{s}}+q\cdot\omega+\frac{(n\cdot\omega)(q\cdot n)}{\sqrt{s}(n^0+\sqrt{s})}\bigg)\frac{(n\cdot\omega)}{\sqrt{s}}\\
\displaystyle
p^k+2\bigg(-q^0\frac{(n\cdot\omega)}{\sqrt{s}}+q\cdot\omega+\frac{(n\cdot\omega)(q\cdot n)}{\sqrt{s}(n^0+\sqrt{s})}\bigg)
\bigg(\omega^k+\frac{(n\cdot\omega)n^k}{\sqrt{s}(n^0+\sqrt{s})}\bigg)
\end{array}
\right),\label{p'_ortho}
\end{align}
and a similar representation for $q'^\alpha$ is obtained. The representations \eqref{Q_ortho} and \eqref{p'_ortho} are generalized to the FLRW case as in \eqref{Boltzmann} and \eqref{p'}.

\section*{Acknowledgements}
H. Lee has been supported by the TJ Park Science Fellowship of POSCO TJ Park Foundation. This research was supported by Basic Science Research Program through the National Research Foundation of Korea(NRF) funded by the Ministry of Science, ICT \& Future Planning(NRF-2015R1C1A1A01055216). E.N. is currently funded by a Juan de la Cierva research fellowship from the Spanish government and this work has been partially supported by ICMAT Severo Ochoa project SEV-2015-0554 (MINECO).  The authors thank the Mathematisches Forschungsinstitut Oberwolfach for the hospitality and support.

\end{document}